\documentclass[10pt,authoryear]{article}
\usepackage{appendix}
\usepackage{setspace}
\usepackage[font=small,  margin=2cm]{caption}
\usepackage[tbtags]{amsmath}
\usepackage{amsthm,amssymb,amsfonts}
\usepackage{natbib}
\usepackage{multirow}
\usepackage{lscape}
\usepackage{subfigure}
\usepackage{makecell}
\usepackage{exscale}
\usepackage{booktabs}
\usepackage{array}
\usepackage{fullpage}
\usepackage{url}
\usepackage{algorithm}
\usepackage{algpseudocode}
\usepackage{bm}
\usepackage{smile}
\usepackage{mathtools}
\usepackage{wrapfig}
\usepackage{lipsum}
\usepackage{threeparttable}
\usepackage{mathrsfs}
\usepackage{relsize}
\usepackage{dsfont}
\usepackage{multirow}
\usepackage{apalike}
\usepackage[usenames,dvipsnames,svgnames,table]{xcolor}
\usepackage[colorlinks=true,linkcolor=blue,urlcolor=blue,citecolor=blue]{hyperref}
\newtheorem{assumption}{Assumption}[section]
\numberwithin{equation}{section}
\numberwithin{theorem}{section}
\numberwithin{corollary}{section}
\numberwithin{definition}{section}
\renewcommand{\baselinestretch}{1.5}

\begin{document}

\title{\LARGE Distributed Learning for Principle Eigenspaces without Moment Constraints}

	\author{
	Yong He
		\footnotemark[1],
	Zichen Liu\footnotemark[1],
Yalin Wang\footnotemark[1]
	}
\renewcommand{\thefootnote}{\fnsymbol{footnote}}
\footnotetext[1]{Institute of Financial Studies, Shandong University, China; e-mail:{\tt heyong@sdu.edu.cn, zhaochangwei@mail.sdu.edu.cn}}

\maketitle

Distributed Principal Component Analysis (PCA) has been studied  to deal with the case when data are stored across multiple machines and communication cost or privacy concerns prohibit the computation of PCA in a central
location. However, the sub-Gaussian assumption in the related  literature is restrictive in real application where outliers or heavy-tailed data are common  in areas such as finance and macroeconomic.
In this article, we propose a distributed algorithm for estimating the principle eigenspaces without any moment constraint on the underlying distribution. We study the problem under the elliptical family framework and adopt the sample multivariate Kendall'tau matrix to extract eigenspace estimators from all sub-machines, which can be viewed as points in the Grassman manifold. We then find the ``center" of these points as the final distributed estimator of the principal eigenspace. We investigate the bias
and variance for the  distributed estimator and derive its  convergence rate which depends  on the effective rank and eigengap of the scatter matrix, and the number of submachines. We show that  the distributed estimator performs as if  we have full access of whole data. Simulation studies show that the distributed algorithm performs comparably with the existing one for light-tailed data, while showing great advantage for heavy-tailed data. We also extend our algorithm to the distributed learning of elliptical factor models and verify its empirical usefulness through  real application to a macroeconomic dataset.
\vspace{2em}


\section{Introduction}
Principal component analysis (PCA) is one of the most important statistical tool for dimension reduction, which extracts latent principal factors that contribute to the most
variation of the data. For the classical setting with fixed dimension, the consistency and asymptotic normality of empirical principal components have been raised since \cite{anderson1963asymptotic}. In the last decade, high-dimensional PCA gradually attracted the attention of statisticians,  see for example \cite{ONATSKI2012244}, \cite{Wei2017Asymp}, \cite{Kong2021Dis}, \cite{bao2022statistical}. The existing work on high-dimensional PCA typically assume the Gaussian/sub-Gaussian tail property of the underlying distribution, which is really an idealization of the complex random real world.
Heavy-tailed
data are common in research fields such as financial
engineering and biomedical imaging \citep{jing2012modeling,kong2015testing,fan2018large,li2022manifold}. Thus it is desperately needed to find a robust way to do principal component analysis for heavy-tailed data.
   Elliptical family provides one a proper way to capture the different tail behaviour of various common distributions such as Gaussian and $t$-distribution, and has been widely studied in various high-dimensional statistical problems, see for example, \cite{Han2014Scale,Yu2019robust,Hu2019High,Chen2021High}.
   \cite{han2018eca} presented a robust alternative to PCA, called Elliptical
Component Analysis (ECA), for analyzing high dimensional, elliptically distributed data, in which multivariate Kendall's tau matrix plays a central role. In essence, for elliptical distributions, the  eigenspace of the population Kendall's tau matrix coincides with that of the scatter matrix
 The multivariate Kendall's tau  is first introduced in \cite{choi1998multivariate} for testing dependence, and is later adopted for estimating covariance matrix  and principal components \citep{taskinen2012robustifying}. Properties of ECA in low dimension were considered by \cite{hallin2010optimal}, \cite{hallin2014efficient}, and  some recent works related to high-dimensional ECA include but not limited to \cite{feng2017high,Chen2021High,he2022large}. Noticeably, \cite{he2022large} extended ECA to factor analysis under the framework of elliptical factor model.

With rapid developments of information and technology, the modern datasets exhibit the characteristic of extremely large-scale and we are now embracing the big-data era. Efficient statistical inference algorithms on such enormous dataset is unprecedentedly desirable.  Distributed computing provides an effective way to deal with  large-scale datasets. In addition to computation efficiency, distributed computing is also relatively robust to possible failures in the subsevers. There are lots of other reasons for establishing rigorous statistical theories on distributed computing, such as privacy protection, data ownerships and limitation of data storage. Two patterns of data segmentation were considered in the literature over the past few years, which are horizontal and vertical. ``Horizontal"  reserves all the features in each subserver, while data are scattered to different machines. Conversely, ``Vertical" means that features are divided into several parts, each storage has full access of data but part of the features, which is common in signal processing and sensor networks. Some representative works  include but not limited to \cite{qu2002principal},  \cite{zhang2012communication}, \cite{fan2019distributed} and  \cite{fan2021communication}. In particular, \cite{fan2019distributed} proposed  a distributed PCA algorithm: first each submachine computes
the top $K$ eigenvectors and transmits them to the central machine; then the central machine aggregates the information from all the submachines and conducts a PCA based on the aggregated information.

In the current work, we consider a robust alternative to the existing distributed PCA algorithm. We adopt the horizontal division and propose a distributed algorithm for estimating the principle eigenspaces without any moment constraint on the underlying distribution. In detail,  assume that we have $m$ subservers under the elliptical family framework, we first extract the leading $K$ orthonormal eigenvectors (as columns of $\hat\bV_K^{(i)}$) of the local sample multivariate Kendall's tau matrix with observations on the $i$-th subserver, for $i=1,\ldots,m$.
  We next transport these eigenspace estimators $\{span(\hat\bV_K^{(i)})\}$ to the central server, where $\{span(\hat\bV_K^{(i)})\}$ denotes the linear space spanned by the columns of $\hat\bV_K^{(i)}$. In essence, the transported local eigenspace estimators $\{span(\hat\bV_K^{(i)})\}$ can be viewed as points (representatives of equivalent class) in the Grassmann manifold, see Figure \ref{manifold} for better illustration. We then find the ``center" ( the red point $span(\widetilde\bV_K)$ in Figure \ref{manifold}) of these points by minimizing the projection metrics on Grassmann manifolds in the central server, which is analogous to the physical notion of barycenter. Mathematically, the ``center" $span(\widetilde\bV_K)$ satisfies
\[
\widetilde\bV_K=\bargmin_{\bV^\top\bV=\Ib}\sum_{i=1}^m\|\bV\bV^\top-\hat\bV_K^{(i)}\hat\bV_K^{(i)\top}\|_F^2,
\]
and it can be shown that $\widetilde\bV_K$ is exactly the leading $K$ eigenvectors of the average projection matrix
$\widetilde \bSigma=1/m\sum_{i=1}^m\hat\bV_K^{(i)}\hat\bV_K^{(i)\top}$.

\begin{figure}[!h]
     \centering \includegraphics[width=16cm]{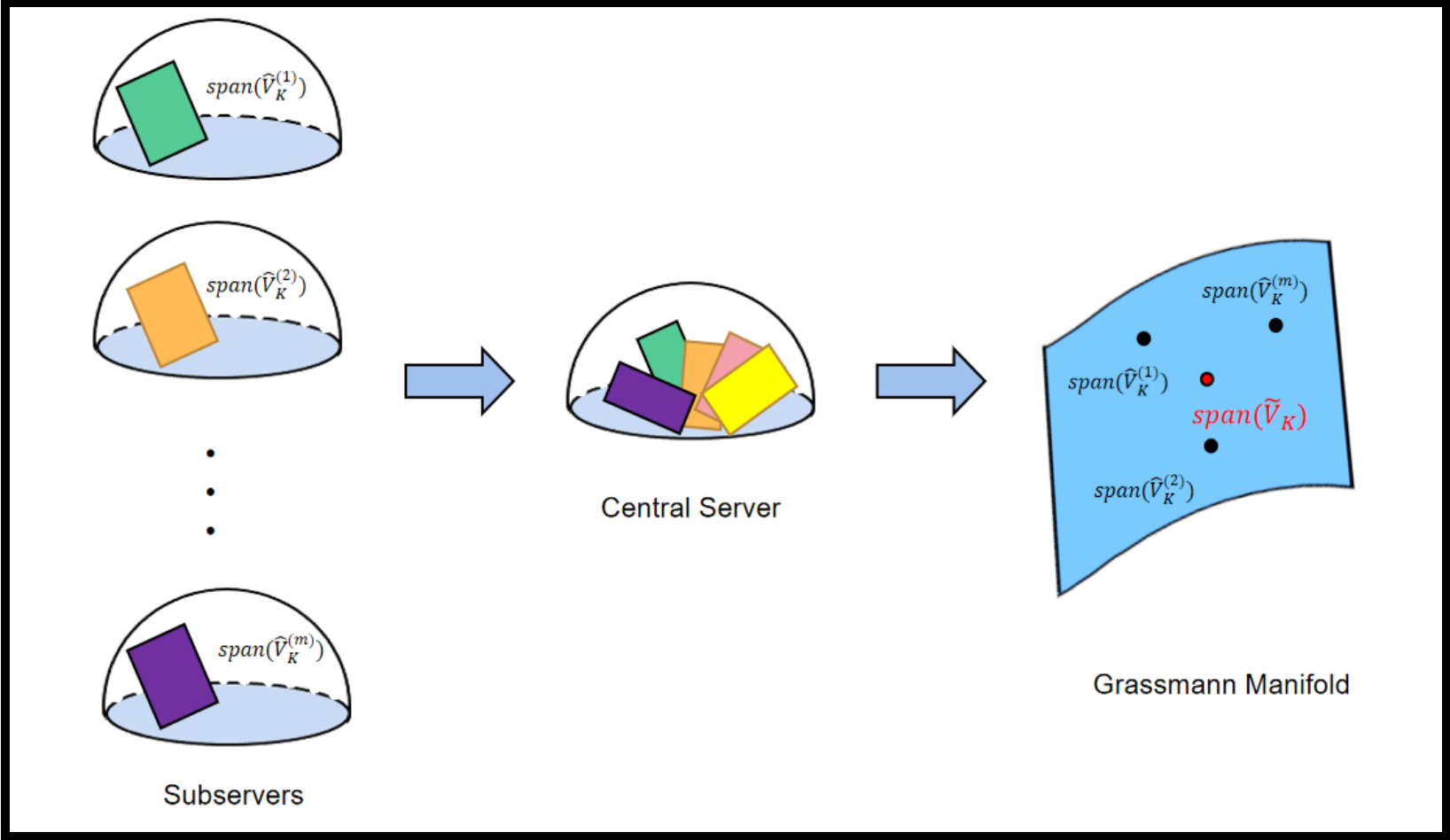}
    \caption{Workflow of the proposed distributed algorithm.}
    \label{manifold}
    \end{figure}

The  contributions of the current work include the following aspects: firstly we for the first time propose a robust alternative to distributed principal PCA. The proposed distributed algorithm is easy to implement and avoids huge transportation costs between central processor and subservers; secondly we investigate the bias and variance for the robust distributed estimator and derive its convergence rate which depends on the effective rank of the scatter matrix, eigengap, and the number of submachines. We show that the distributed estimator performs as if we
have full access of whole data; thirdly,  we extend the  distributed algorithm to the elliptical factor model in a horizontal partition
regime, which is the first distributed algorithm for robust factor analysis as far as we know.

The rest of this paper is organized as follows. Section \ref{section2} includes some notations and preliminary results on elliptical family and multivariate Kendall'tau matrix. In Section \ref{sec:new}, we present the details on our robust distributed algorithm for estimating principal eigenspace. In Section \ref{section3}, we investigate the theoretical  properties of the robust distributed estimator. In Section \ref{section5}, we extend the distributed algorithm to the  distributed elliptical factor analysis. Simulation results are given in Section \ref{section4}. We also applied the distributed algorithm to analyze a large-scale real macroeconomic dataset. Further discussions and future work direction are left in Section \ref{section6}.

\section{Notations and Preliminaries}\label{section2}
We first introduce some notations used throughout the article. Vectors and matrices will be written in bold symbol, scalars are written by regular letters. For two sequences $\{a_{n}\}_{n=1}^{\infty},\{b_{n}\}_{n=1}^{\infty}$, $a_{n}\lesssim b_{n}$ if there is a universal constant $C>0$, such that $a_{n}\leq Cb_{n}$, for which we also adopt another notation $a_{n}=O(b_{n})$, while $a_{n}=o(b_{n})$ means $a_{n}/b_{n}\rightarrow 0$, as $n\rightarrow \infty$. Index set $\{1,2,\ldots,p\}$ is simply denoted by $[p]$. For vectors, we define $\boldsymbol{e}_{i}$ to be the unit vector with 1 in the $i$-th component, $\Vert \cdot \Vert_{p}$ is the $\ell_{p}$-norm. For random variable $X$, Orlicz $\psi_1$ norm $\Vert X \Vert_{\psi_1}$ is defined as $\mathop{sup}_{p\geq1}(\mathbb{E}\lvert X \rvert^{p})^{1/p}/p$. If random vectors $\boldsymbol{X},\boldsymbol{Y}$ share the same distribution, we write $\boldsymbol{X}\overset{d}{=}\boldsymbol{Y}$. For matrix $\boldsymbol{A}\in\mathbb{R}^{p\times p}$, its transpose is $\boldsymbol{A}^\top$; $\lambda_{i}(\boldsymbol{A})$ is the $i$-th largest eigenvalue of $\boldsymbol{A}$; $\Vert \boldsymbol{A} \Vert_{2}$ is the spectral norm and $\Vert \boldsymbol{A} \Vert_{F}$ is the Frobenius norm. Given a fixed $K\in[p]$, the eigen gap of $\boldsymbol{A}$ is $\Delta(\boldsymbol{A})=\lambda_{K}(\boldsymbol{A})-\lambda_{K+1}(\boldsymbol{A})$, and $r=r(\boldsymbol{A}):=Tr(\boldsymbol{A})/\lambda_1(\boldsymbol{A})$ is the effective rank of $\boldsymbol{A}$. The space spanned by the columns of $\boldsymbol{A}$ is denoted by $span(\boldsymbol{A})$. For two matrices with orthogonal columns $\boldsymbol{A}\in\mathbb{R}^{p\times n_1},\boldsymbol{B}\in\mathbb{R}^{p\times n_2}$($n_1,n_2\leq p$), the spatial distance between $span(\boldsymbol{A})$ and $span(\boldsymbol{B})$ is $\rho(\boldsymbol{A},\boldsymbol{B})=\Vert \boldsymbol{A}\boldsymbol{A}^\top-\boldsymbol{B}\boldsymbol{B}^\top \Vert_{F}$.

\vspace{0.5em}
\noindent\textbf{Elliptical Distribution and Spatial Kendall's Tau Matrix}
\vspace{0.5em}

Elliptical family is a large distribution family containing common distributions such as Gaussian and $t$-distribution, which is widely used for modeling heavy-tailed data in finance and macroeconomics. To begin with, we recall its  definition and some of its nice properties. For further details on the elliptical distribution, see for example \cite{hult2002multivariate}.

In the following we give two equivalent definitions, one  by its characteristic function, the other by its stochastic representation. A random vector $\boldsymbol{X}=(X_{1},\ldots,X_{p})^\top$ belongs to the elliptical family,  denoted by $\boldsymbol{X}\sim ED_p(g;\boldsymbol{\mu},\boldsymbol{\Sigma})$, if its characteristic function has the form:
\[
\textbf{t}\mapsto \phi_{g}(\textbf{t};\boldsymbol{\mu},\boldsymbol{\Sigma})=e^{i\textbf{t}^\top\boldsymbol{\mu}}g(\textbf{t}^\top\boldsymbol{\Sigma} \textbf{t}),
\]
 where $g(\cdot)$ is a proper function defined on $[0,\infty)$, $\boldsymbol{\mu}\in\mathbb{R}^{p}$ is the location parameter, $\boldsymbol{\Sigma}$ is called scatter matrix with $rank(\boldsymbol{\Sigma})=q$. Equivalently,
\[
\boldsymbol{X}\overset{d}{=}\boldsymbol{\mu}+\xi\boldsymbol{AU},
\]
where $\boldsymbol{U}$ is a random vector uniformly distributed on the unit sphere $S^{q-1}$ in $\mathbb{R}^{q}$, $\xi \geq 0$ is a scalar random variable indepent of $\boldsymbol{U}$, $\boldsymbol{A}\in\mathbb{R}^{p\times q}$ is a deterministic matrix satisfying $\boldsymbol{AA}^\top=\boldsymbol{\Sigma}$. Thus we can also denote as $\boldsymbol{X}\sim ED_p(\boldsymbol{\mu},\boldsymbol{\Sigma},\xi)$ and $\xi$ and $g(\cdot)$ are mutually determined.
 Elliptically distributed random vectors have the same nice properties as Gaussian random vectors, such as  linear combination of elliptically distributed random vectors also follows elliptical distribution and their marginal distributions are also elliptical.

 For Gaussian distribution, scatter matrix is  the covariance matrix up to a  constant. For non-Gaussian distribution, even with infinite variance (such as Cauchy distribution), scatter matrix still measures the dispersion of a random vector. To estimate the principal eigenspace of elliptically distributed data, simply performing PCA to the sample covariance matrix is not satisfactory due to the possible inexistence of the population covariance matrix. To address the problem, we turn to a tool,  the spatial Kendall's tau matrix.
For $\boldsymbol{X} \sim ED(\boldsymbol{\mu},\boldsymbol{\Sigma},\xi)$ and its independent copy $\widetilde{\boldsymbol{X}}$, the population spatial Kendall's tau matrix of $\boldsymbol{X}$ is defined as:
\[
\boldsymbol{K}=\mathbb{E}\left\{\frac{(\boldsymbol{X}-\widetilde{\boldsymbol{X}})(\boldsymbol{X}-\widetilde{\boldsymbol{X}})^\top}{\Vert \boldsymbol{X}-\widetilde{\boldsymbol{X}}\Vert_{2}^{2}}\right\}.
\]
The spatial Kendall's tau matrix was first introduced  in \cite{choi1998multivariate} and has been used for a lot of statistical problems such as covariance matrix estimation and principal eigenspace estimation, see \cite{visuri2000sign}, \cite{fan2018large} and \cite{han2018eca}. A critical property of spatial Kendall's tau matrix  is that for elliptical vector $\bX$, it shares the same eigenvectors of scatter matrix with the same ordering of eigenvalues, see \cite{han2018eca} for detailed proof.
Suppose $\left\{ \boldsymbol{X}_{t}\right\}_{t=1}^{n}$ are $n$ independent data points from $\bX$, the sample spatial Kendall's tau matrix is naturally defined as a U-statistic:
\[
\widehat{\boldsymbol{K}}=\frac{2}{n(n-1)}\sum_{i<j}\frac{(\boldsymbol{X}_{i}-\boldsymbol{X}_{j})(\boldsymbol{X}_{i}-\boldsymbol{X}_{j})^\top}{\Vert \boldsymbol{X}_{i}-\boldsymbol{X}_{j}\Vert_{2}^{2}}.
\]

\section{Distributed PCA without Moment Constraint}\label{sec:new}
In this section, a new distributed algorithm for estimating principal eigenspace robustly is introduced in detail, which generalizes the distributed PCA by \cite{fan2019distributed} to the elliptical family setting.

Suppose there are $N$ samples of $\bX\sim ED_p(\boldsymbol{\mu},\boldsymbol{\Sigma},\xi)$ in total, and we are interested in recovering the eigenspace spanned by the leading $K$ eigenvectors of the scatter matrix $\bSigma$. The $N$ samples  spread across $m$ machines, and
on the $l$-th machine, there are $n=N/m$ samples. For $l\in[m]$, let $\{ \boldsymbol{X}_{t}^{(l)} \}_{t=1}^{n}$ denote the samples stored on the $l$-th machine.  First calculate the local sample spatial Kendall' tau matrix on the $l$-th server as
 $$\widehat{\boldsymbol{K}}^{(l)}=\frac{1}{C_{n}^{2}}\sum_{j<s}\frac{(\boldsymbol{X}_{j}^{(l)}-\boldsymbol{X}_{s}^{(l)})(\boldsymbol{X}_{j}^{(l)}-\boldsymbol{X}_{s}^{(l)})^\top}{\Vert \boldsymbol{X}_{j}^{(l)}-\boldsymbol{X}_{s}^{(l)} \Vert_{2}^{2}},  \ \ l=1,\ldots,m.$$
We further compute the  leading $K$ eigenvectors of sample Kendall's tau matrix $\widehat{\boldsymbol{K}}^{(l)}$ on the $l$-th server and denote $\widehat{\boldsymbol{V}}_{K}^{(i)}$ as the matrix with columns being these eigenvectors.
We next transport these $\widehat{\boldsymbol{V}}_{K}^{(i)}$ to the central server. The communication cost of the proposed distributed algorithm is of order $O(mKp)$. In contrast, to share all the data or entire spatial Kendall'tau matrix,
the communication cost will be of order $O(mp \min(n, p))$. In most cases
$K = o(min(n, p))$, the distributed algorithm requires much less communication cost
than naive data aggregation.

The transported local  estimator  $\widehat{\boldsymbol{V}}_{K}^{(i)}$ would span an eigenspace $\{span(\hat\bV_K^{(i)})\}$, which can be viewed as points  in the Grassmann manifold, as illustrated in Figure \ref{manifold}. We then find the ``center"  of these points by minimizing the projection metrics on Grassmann manifolds in the central server, i.e,
\[
\widetilde\bV_K=\bargmin_{\bV^\top\bV=\Ib}\sum_{i=1}^m\|\bV\bV^\top-\hat\bV_K^{(i)}\hat\bV_K^{(i)\top}\|_F^2.
\]
Denote $P_{\bV}=\bV\bV^\top$ and $P_{\hat\bV_K^{(i)}}=\hat\bV_K^{(i)}\hat\bV_K^{(i)\top}$, which are the projection matrices of $span(\bV)$ and $span(\hat\bV_K^{(i)})$ respectively.  Then we have:
\begin{equation*}
	\sum_{i}^m\|\bV\bV^\top-\hat\bV_K^{(i)}\hat\bV_K^{(i)\top}\|^2_F=\sum_{i=1}^m Tr(P_{\bV})+\sum_{i=1}^m Tr(P_{\hat\bV_K^{(i)}})-2Tr\left[P_{\bV}(\sum_{i=1}^m P_{\hat\bV_K^{(i)}})\right].
\end{equation*}

\begin{algorithm}
\caption{Robust Distributed Principle Component Analysis Algorithm}
\label{alg:A}%方便后续引用
\begin{algorithmic}
\item
\textbf{STEP 1}:\par
On the $l$-th machine, compute the $K$ leading eigenvectors $\widehat{\boldsymbol{V}}_{K}^{(l)}=(\widehat{\boldsymbol{v}}_{1}^{(l)},\ldots,\widehat{\boldsymbol{v}}_{K}^{(l)})_{p\times K}$ of the sample Kendall's tau matrix $\widehat{\boldsymbol{K}}^{(l)}$;\\
\\
\textbf{STEP 2}:\par
On the central server, compute the $K$ leading eigenvectors $\widetilde{\boldsymbol{V}}_{K}=(\widetilde{\boldsymbol{v}}_{1},\ldots,\widetilde{\boldsymbol{v}}_{K})$ of $\widetilde{\boldsymbol{\Sigma}}=\frac{1}{m}\sum_{i=1}^{m}\widehat{\boldsymbol{V}}_{K}^{(i)}\widehat{\boldsymbol{V}}_{K}^{(i)^\top}$ \\
\\
\textbf{Output}:\par
The estimator of the leading eigenspace, $span(\widetilde{\boldsymbol{V}}_{K})$

\end{algorithmic}
\end{algorithm}

The first two terms on the right hand side are fixed, so it's equivalent to maximizing the last term, namely $\mathcal{L}(\bV) = \tr\left[\bV^{\top}(\sum_{i=1}^m P_{\hat\bV_K^{(i)}})\bV\right]$. Thus it's clear that $\widetilde\bV_K$ is the leading $K$ eigenvectors of the average projection matrix $$\widetilde\bSigma= \frac{1}{m}\sum_{i=1}^m P_{\hat\bV_K^{(i)}}=\frac{1}{m}\sum_{i=1}^{m}\widehat{\boldsymbol{V}}_{K}^{(i)}\widehat{\boldsymbol{V}}_{K}^{(i)^\top}.$$
At last we take $\widetilde\bV_K$ as the final distributed estimator of the principal eigenspace.
The detailed algorithm
 is summarized in Algorithm \ref{alg:A}.

 \cite{fan2019distributed} proposed a distributed PCA method in which they explain the aggregation procedure from the perspective of semidefinite programming (SDP) problem and minimize the sum of squared losses.  In the current work, we provide a rationale of the aggregation procedure from the perspective of the Grassmann manifold, which is quite easy to understand. The viewpoint from the Grassmann manifold is similar to \cite{li2022manifold}, where they proposed manifold PCA for matrix elliptical factor model but not from the perspective of distributed computation.

\section{Theoretical properties}\label{section3}
In this section, we investigate the  theoretical property of the distributed estimator $\widetilde \bV_K$. We focus on the distance of the spaces spanned by the columns of $\widetilde{\boldsymbol{V}}_{K}$ and those of $\boldsymbol{V}_{K}$.  We adopt the metric in \cite{fan2019distributed} to analyse the statistical error of $span\{\widetilde{\boldsymbol{V}}_{K}\}$ as an estimator of $span\{\boldsymbol{V}_{K}\}$. At the beginning, we give some assumptions for technical analysis.
\begin{assumption}\label{assump1}
Assume that $\bX\sim ED_p(\bmu,\bSigma,\xi)$, for a fixed $\delta >0$, $0<\alpha <1/2$, the scatter matrix $\boldsymbol{\Sigma}$ satisfies:
\begin{itemize}
\item [\romannumeral1)] There exists a universal positive constant $C>0$ such that the ratio $\gamma$ of the maximal eigenvalue to minimal eigenvalue of the scatter matrix satisfies:
\[
\gamma =\frac{\lambda_{1}(\boldsymbol{\Sigma})}{\lambda_{p}(\boldsymbol{\Sigma})} \leq Cp^{\alpha};
\]

\item [\romannumeral2)] There is a relatively significant difference between the $K$-th eigenvalue and the ($K+1$)-th eigenvalue:
\[
\frac{\lambda_{K}(\boldsymbol{\Sigma})}{\lambda_{K+1}(\boldsymbol{\Sigma})} \geq 1+\delta
\]
\end{itemize}

Assumption \ref{assump1} imposes conditions on the eigengap and condition number of the scatter matrix, which is common in the high-dimensional PCA literature and is mild in the following sense. In terms of  elliptical factor models in \cite{he2022large}, that's $\bX_t=\bL\bbf_t+\bu_t$,
 the scatter matrix of $\bX_t$ has the low rank plus sparsity structure $\boldsymbol{\Sigma}=\boldsymbol{L}\boldsymbol{L}^\top+\boldsymbol{\Sigma}_{u}$, where $\boldsymbol{\Sigma}_{u}$  is the scatter matrix of idiosyncratic errors $\bu_t$, see also Section \ref{section5}.  Assume that there exists $K$ latent factors, the loading matrix $\boldsymbol{L}$ satisfies the pervasive condition $\boldsymbol{L}^\top\boldsymbol{L}/p^{\alpha}\rightarrow \boldsymbol{Q}$ as $p\rightarrow \infty$, with $\boldsymbol{Q}$ being a $K\times K$ positive definite matrix, and  the eigenvalues of $\bSigma_{u}$ are bounded from below and above, then the scatter matrix $\bSigma$  would satisfy the conditions in Assumption \ref{assump1}.
\end{assumption}

In the following analysis,
we decompose the error $\rho(\widetilde{\boldsymbol{V}}_{K},\boldsymbol{V}_{K})$ into two parts: the bias term and the variance term. In detail,
$$
\rho(\widetilde{\boldsymbol{V}}_{K},\boldsymbol{V}_{K}) \leq \underbrace{\rho(\widetilde{\boldsymbol{V}}_{K},\boldsymbol{V}_{K}^{*})}_{\text{variance term}}+\underbrace{\rho(\boldsymbol{V}_{K}^{*}, \boldsymbol{V}_{K})}_{\text{bias term}},
$$
where $\boldsymbol{V}_{K}^{*}$ is the $K$ leading eigenvectors of $\boldsymbol{\Sigma}^{*}=\mathbb{E}(\widehat{\boldsymbol{V}}_{K}^{(i)}\widehat{\boldsymbol{V}}_{K}^{(i)^\top})$. In the following we present the upper bound of the variance term  and the  bias term  in Theorem \ref{theo:1} and Theorem \ref{theo:2} respectively.

\begin{theorem}\label{theo:1}
Under Assumption \ref{assump1}, further assume that $\log p=o(N)$. Then we have
$$
\Vert {\rho}(\widetilde{\boldsymbol{V}}_{K},\boldsymbol{V}_{K}^{*})\Vert_{\psi_1}\lesssim \kappa rK^{3/2}\sqrt{\frac{\log p}{N}},
$$
where \[
\kappa =\frac{\lambda_{1}(\boldsymbol{\Sigma})}{\lambda_{K}(\boldsymbol{\Sigma})-\lambda_{K+1}(\boldsymbol{\Sigma})}, \ \ r=\frac{Tr(\boldsymbol{\Sigma})}{\lambda_{1}(\boldsymbol{\Sigma})}.
\]
\end{theorem}
The measure of concentration  is  essential in high-dimensional
statistical analysis. \cite{fan2019distributed} exerted the sub-Gaussian assumption to derive a form of concentration inequality. However,
the proposed algorithm maps each eigenspace estimator to the Grassmann manifolds where concentration inequalities are easily derived due to the compactness therein, thus the results also hold even for Cauchy distribution without any moment.
Next we give the upper bound of the bias term.

\begin{theorem}\label{theo:2}
Under Assumption \ref{assump1}, further assume that $\log p=o(N/m)$. Then we have
$$
{\rho}(\boldsymbol{V}_{K}^{*},\boldsymbol{V}_{K})\lesssim \kappa ^{2}r^{2}K^{5/2}\frac{m\log p}{N}.
$$
\end{theorem}

By combining the results of the bias term and the variance term, we have the following corollary.
\begin{corollary}\label{coro:1}
Under Assumption \ref{assump1}, further assume that $\log p=o(N/m)$. Then there exist positive constants $C_{1},C_{2}$ such that
$$\Vert {\rho}(\widetilde{\boldsymbol{V}}_{K},\boldsymbol{V}_{K})\Vert_{\psi_1}\leq C_{1}\kappa rK^{3/2}(\frac{\log\,p}{N})^{1/2}+C_{2}\kappa ^{2}r^{2}K^{5/2}\frac{m\log p}{N};$$
if the number of subservers satisfies $m\leq C(N/\log p)^{1/2}/(\kappa rK)$ for some constant $C$, we have
$$\Vert {\rho}(\widetilde{\boldsymbol{V}}_{K},\boldsymbol{V}_{K})\Vert_{\psi_1}\lesssim \kappa rK^{3/2}\left(\frac{\log p}{N}\right)^{1/2}.$$
\end{corollary}

By Corollary \ref{coro:1} and the theoretical result in \cite{han2018eca}, we conclude that our algorithm is not only  communication-efficient  but also outputs an estimator which shares the same convergence rate as the ECA estimator with all samples available at one machine  when $m$ is not larger than the bounds given in Corollary \ref{coro:1}. This is also confirmed  in the following simulation study.

\section{Extension to Elliptical factor models}\label{section5}
In this section, we assume that the data points $\bX_t$ has a elliptical factor model, i.e.,
$\left\{\boldsymbol{X}_{t}\right\}_{t=1}^{N}\subset\mathbb{R}^{p}$,
\[
\boldsymbol{X}_{t}=\boldsymbol{L}\boldsymbol{f}_{t}+\boldsymbol{u}_{t}, \ \ t=1,\ldots, N,
\]
where $\boldsymbol{X}_{t}=(X_{1t},\ldots,X_{pt})^\top, \boldsymbol{f}_{t}\in\mathbb{R}^{K}$ are the latent factors, $\boldsymbol{L}=(\boldsymbol{l}_{1},\ldots,\boldsymbol{l}_{p})^\top$ is the factor loading matrix and $\boldsymbol{u}_{t}$'s are the idiosyncratic errors. Note that for the elliptical factor model, the loading space $span(\bL)$ rather than $\bL$ is identifiable. We assume that random vectors $(\boldsymbol{f}_{t}^\top,\boldsymbol{u}_{t}^\top)^\top$ are generated independently and identically from the elliptical distribution $ED_{K+p}(\zero,\bSigma_0,\xi)$, where $\bSigma_0=\text{diag}(\Ib_{K},\bSigma_u)$, see \cite{he2022large} for more details on elliptic distribution. By the property of the elliptical family, the scatter matrix of $\boldsymbol{X}_{t}$ has the form: $\boldsymbol{\Sigma}=\boldsymbol{L}\boldsymbol{L}^\top+\boldsymbol{\Sigma}_{u}$. We assume that the loading matrix $\boldsymbol{L}$ satisfies the pervasive condition $\boldsymbol{L}^\top\boldsymbol{L}/p^{\alpha}\rightarrow \boldsymbol{Q}$ as $p\rightarrow \infty$, where $\boldsymbol{Q}$ is a $K\times K$ positive definite matrix, and  the eigenvalues of $\bSigma_{u}$ are bounded from below and above. Then $\bSigma$  satisfies Assumption \ref{assump1}. Naturally, we use the space spanned by the $K$ leading eigenvectors of the sample Multivariate Kendall's tau matrix as the estimator of the loading space $span(\bL)$. In the article, we assume that the $N$ observations spread across $m$ machines, and
on the $l$-th machine, there are $n=N/m$ observations. For $l\in[m]$, let $\{ \boldsymbol{X}_{t}^{(l)} \}_{t=1}^{n}$ denote the observations stored on the $l$-th machine.
\[
\boldsymbol{X}_{t}^{(l)}=\boldsymbol{L}\boldsymbol{f}_{t}^{(l)}+\boldsymbol{u}_{t}^{(l)}, \ \ t=1,\ldots, n, \ \ l=1\ldots, m.
\]
\begin{algorithm}\label{method2}
\caption{Distributed algorithm for elliptical factor model}
\label{alg:2}
\begin{algorithmic}
\item
\textbf{STEP1}:\par
On the $l$-th machine, $l=1,\ldots,m$, compute the $K$ leading eigenvectors $\widehat{\boldsymbol{V}}_{K}^{(l)}=(\widehat{\boldsymbol{v}}_{1}^{(l)},\ldots,\widehat{\boldsymbol{v}}_{K}^{(l)})_{p\times K}$ of the sample Kendall's tau matrix, and transport $\widehat{\boldsymbol{V}}_K^{(l)},l=1\ldots,m$ to the central server.\\
\\
\textbf{STEP2}:\par
On the central processor, compute the $K$ leading eigenvectors $\widetilde{\boldsymbol{V}}_{K}=(\widetilde{\boldsymbol{v}}_{1},\ldots,\widetilde{\boldsymbol{v}}_{K})$ of $\widetilde{\boldsymbol{\Sigma}}=\frac{1}{m}\sum_{l=1}^{m}\widehat{\boldsymbol{V}}_{K}^{(l)}\widehat{\boldsymbol{V}}_{K}^{(l)^\top}$ and use $span(\widetilde{\boldsymbol{V}}_{K})$ as the estimator of factor loading space $span(\boldsymbol{L})$.\\
\\
\textbf{STEP3}:\par
Transport $\widetilde{\boldsymbol{V}}_{K}$ back to each servers, and set $\tilde \bL=p^{\alpha/2}\widetilde{\boldsymbol{V}}_{K}$, further solve the least squares problem on each server to get estimators of the factor scores, $\widetilde{\boldsymbol{f}}_{t}^{(l)}$.
\[
\begin{aligned}
\widetilde{\boldsymbol{f}}_{t}^{(l)}=\mathop{\arg\min}\limits_{\boldsymbol{\beta}\in \mathbb{R}^{K}}\sum_{i=1}^{p}(X_{it}-\widetilde{\boldsymbol{l}}_{i}^\top\boldsymbol{\beta}_{j})^{2}=\mathop{\arg\min}\limits_{\boldsymbol{\beta}\in\mathbb{R}^{K}}\Vert \boldsymbol{X}_{t}-\widetilde{\boldsymbol{L}}\boldsymbol{\beta}_{j}\Vert_{2}^{2},
\end{aligned}
\]\\
where  $\widetilde \bL=(\widetilde{\boldsymbol{l}}_{1},\ldots,\widetilde{\boldsymbol{l}}_{p})^\top$.
\\
\textbf{Output}:\par
Factor loading space estimator $span(\widetilde{\boldsymbol{V}}_{K})$, factor scores estimators $\widetilde{\boldsymbol{f}}_{t}^{(l)}, l=1,\ldots,m, t=1,\ldots,n$.
\end{algorithmic}
\end{algorithm}
By the proposed distributed algorithm, we further propose a distributed  procedure to estimate the factor loading space $span(L)$ and the factor scores. The algorithm goes as follows: firstly for $l=1,\ldots,m$, compute the $K$ leading eigenvectors $\{\widehat{\boldsymbol{v}}_{1}^{(l)},\ldots,\widehat{\boldsymbol{v}}_{K}^{(l)}\}$ of the sample Kendall's tau matrix on the $l$-th server, and let $\widehat{\boldsymbol{V}}_{K}^{(l)}=(\widehat{\boldsymbol{v}}_{1}^{(l)},\ldots,\widehat{\boldsymbol{v}}_{K}^{(l)})_{p\times K}$;  then transport $\widehat{\boldsymbol{V}}_{K}^{(l)}$ to the central processor; then perform spectral decomposition for the average of projection operators $\widetilde\bSigma={1}/{m}\sum_{l=1}^{m}\widehat{\boldsymbol{V}}_{K}^{(l)}\widehat{\boldsymbol{V}}_{K}^{(l)^\top}$ to get its $K$ leading eigenvectors $\{\widetilde{\boldsymbol{v}}_{1},\ldots,\widetilde{\boldsymbol{v}}_{K}\}$, and let $\widetilde{\boldsymbol{V}}_{K}=(\widetilde{\boldsymbol{v}}_{1},\ldots,\widetilde{\boldsymbol{v}}_{K})$. At last $span(\widetilde{\boldsymbol{V}}_{K})$ is set as the final estimator of $span(\bL)$.
If we are also interested in the factor scores on each server, then we simply transport $\widetilde{\boldsymbol{V}}_{K}$ back to each subserver,
 and solve a least squares problem to get the estimators of the factor scores. We summarized the procedure in Algorithm \ref{alg:2}.

\section{Simulation Studies and Real Example}\label{section4}
\subsection{Simulation Studies}
In this section, we investigate the empirical performance of the proposed algorithm in terms of  estimating loading space in factor models. In essence, our algorithm is a distributed Elliptical Component Analysis procedure, thus is briefly denoted as D-ECA.
We compare ours with the distributed PCA  (denoted as D-PCA) algorithm in \cite{fan2019distributed} as well as an oracle algorithm  in which the full samples are available at one machine and elliptical component analysis (denoted as F-ECA) is performed.

We adopt a simplified  data-generating scheme from the elliptical factor model, similar as that in \cite{he2022large}:
\[
X_{it}=\sum_{j=1}^{K}L_{ij}f_{jt}+u_{it},
\]
where $\bbf_t=(f_{1t},\ldots,f_{Kt})^\top$ and $\bu_t=(u_{1t},\ldots,u_{pt})^\top$ are jointly generated from  elliptical distributions. We set $K=3$ and draw $L_{ij}$  independently  from the standard normal distribution.
We use a fixed number of observations on each server, i.e., $n=200$, but vary the number of servers $m$ and the dimensionality  $p$. The factors and the idiosyncratic errors  $(\bbf_t^\top,\bu_t^\top)^\top$ are generated in the following ways: (i) $(\bbf_t^\top,\bu_t^\top)^\top$ are \emph{i.i.d.} jointly elliptical random variables from multivariate Gaussian distributions $\mathcal{N}(\zero,\Ib_{p+K})$; (ii) $(\bbf_t^\top,\bu_t^\top)^\top$ are \emph{i.i.d.} jointly elliptical random variables from  multivariate centralized $t$ distributions $t_{\nu}(\zero,\Ib_{p+K})$ with $\nu=1, 2, 3$.

We compare the performance of the distributed algorithms by  the distance  between the estimated loading space and the real loading space. We adopt a distance between linear spaces also utilized in \cite{he2022large} Precisely, for two matrices with orthonormal columns $\boldsymbol{V}_{1}\in\mathbb{R}^{p\times K}$ and $\boldsymbol{V}_{2}\in\mathbb{R}^{p\times K}$, the distance between their column spaces is defined as:
$$\rho_{1}(\boldsymbol{V}_{1},\boldsymbol{V}_{2})=\left(1-\frac{1}{K}Tr(\boldsymbol{V}_{1}\boldsymbol{V}_{1}^\top\boldsymbol{V}_{2}\boldsymbol{V}_{2}^\top)\right)^{1/2}.$$
The measure equals to 0 if and only if the column spaces of $\boldsymbol{V}_{1}$ and $\boldsymbol{V}_{2}$ are the same, and equals to 1 if and only if they are orthogonal. Indeed, we have $\rho(\boldsymbol{V}_{1},\boldsymbol{V}_{2})=\sqrt{2K}\rho_{1}(\boldsymbol{V}_{1},\boldsymbol{V}_{2})$.

\begin{figure}[!h]
  \centering{\includegraphics[width=16cm]{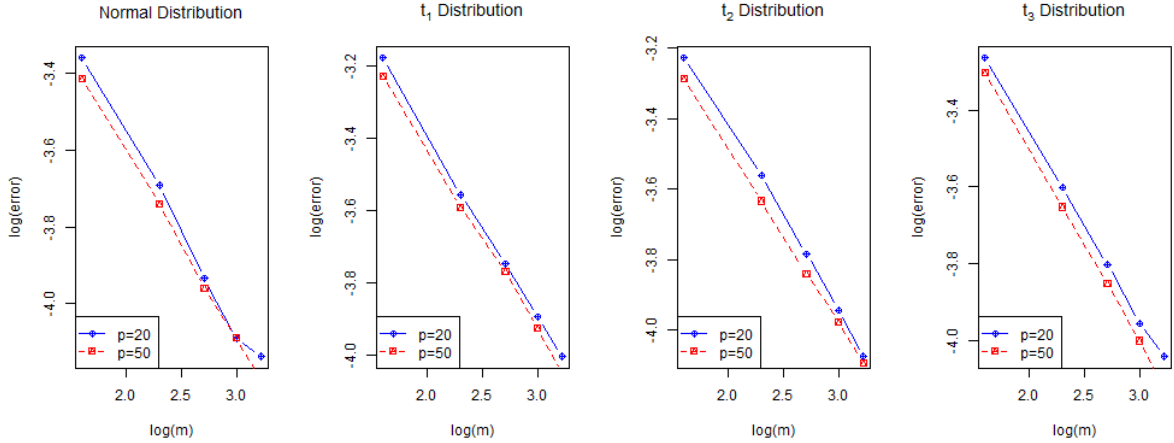}}
 \caption{The log errors, i.e., $\log(\rho_1(\bL,\widehat\bL))$  versus the number of servers $m$ under various distributions with $p=20$ and $p=50$. }\label{fig:1}
 \end{figure}

\begin{table}[!h]
\caption{ Simulation results for the different scenarios based on 100 replications, with
standard deviations in parentheses.}\label{table:1}
\renewcommand{\arraystretch}{1.6}
\centering
\scalebox{0.8}{\begin{tabular}{ccccccc}
\toprule[2pt]
$p$&$m$&$Method$&\multicolumn{4}{c}{$Distribution$} \\
\cline{4-7}
&&&$\mathcal{N}(\boldsymbol{0},\boldsymbol{I}_{p+m})$ &$t_{3}(\boldsymbol{0},\boldsymbol{I}_{p+m})$ &$t_{2}(\boldsymbol{0},\boldsymbol{I}_{p+m})$ &$t_{1}(\boldsymbol{0},\boldsymbol{I}_{p+m})$\\
\midrule
\multirow{9}{*}{20}
&\multirow{3}{*}{5}&D-PCA    &0.034(0.006)   &0.080(0.019)      &0.126(0.034)      &0.259(0.066) \\
                   &&D-ECA   &0.035(0.006)     &0.038(0.006)       &0.040(0.007)      &0.042(0.007) \\
                   &&F-ECA   &0.034(0.005)   &0.038(0.006)          &0.039(0.006)         &0.041(0.007) \\
&\multirow{3}{*}{10}  &D-PCA    &0.024(0.005)      &0.057(0.013)          &0.092(0.022)          &0.169(0.031) \\
                      &&D-ECA  &0.025(0.005)        &0.027(0.005)          &0.028(0.004)          &0.029(0.004) \\
                      &&F-ECA   &0.025(0.005)        &0.027(0.005)          &0.028(0.004)          &0.028(0.004) \\
&\multirow{3}{*}{20}  &D-PCA    &0.016(0.002)        &0.040(0.008)          &0.064(0.013)          &0.124(0.026) \\
                       &&D-ECA  &0.017(0.002)        &0.019(0.008)           &0.019(0.003)        &0.020(0.004) \\
                       &&F-ECA   &0.017(0.002)         &0.019(0.008)          &0.019(0.003)          &0.020(0.004) \\
\hline
\multirow{9}{*}{50}
&\multirow{3}{*}{5}     &D-PCA    &0.032(0.003)    &0.076(0.014)      &0.126(0.026)      &0.247(0.038) \\
                         &&D-ECA  &0.033(0.003)    &0.037(0.003)      &0.037(0.003)      &0.040(0.003) \\
                         &&F-ECA   &0.033(0.003)   &0.036(0.003)      &0.037(0.003)       &0.040(0.003) \\
 &\multirow{3}{*}{10}&D-PCA    &0.023(0.002)        &0.054(0.008)          &0.085(0.012)          &0.167(0.023) \\
                     &&D-ECA  &0.023(0.002)         &0.026(0.002)          &0.026(0.002)          &0.028(0.002) \\
                     &&F-ECA   &0.023(0.002)        &0.026(0.002)          &0.026(0.002)          &0.027(0.002)\\
&\multirow{3}{*}{20} &D-PCA    &0.016(0.002)       &0.038(0.004)          &0.060(0.008)          &0.116(0.013) \\
                      &&D-ECA  &0.017(0.002)        &0.018(0.002)        &0.019(0.002)          &0.020(0.002) \\
                      &&F-ECA   &0.017(0.002)        &0.018(0.002)          &0.019(0.002)          &0.020(0.002) \\
\hline
\multirow{9}{*}{100}
&\multirow{3}{*}{5} &D-PCA    &0.032(0.002)        &0.077(0.014)       &0.123(0.020)      &0.240(0.031) \\
                    &&D-ECA  &0.033(0.002)         &0.036(0.002)       &0.037(0.002)      &0.039(0.002)\\
                    &&F-ECA   &0.033(0.002)        &0.036(0.002)       &0.036(0.002)          &0.039(0.002) \\
&\multirow{3}{*}{10} &D-PCA    &0.023(0.001)    &0.054( 0.007)      &0.087(0.011)      &0.165(0.018) \\
                     &&D-ECA  &0.023(0.001)        &0.026(0.001)          &0.026(0.001)         &0.028(0.002) \\
                     &&F-ECA   &0.023(0.001)        &0.025(0.001)          &0.026(0.001)        &0.028(0.002)  \\
&\multirow{3}{*}{20}&D-PCA &0.016(0.001)      &0.037(0.003)     &0.059(0.006)      &0.116(0.010) \\
                     &&D-ECA  &0.017(0.001)        &0.018(0.001)          &0.019(0.001)          &0.020(0.001) \\
                     &&F-ECA   &0.016(0.001)        &0.018(0.001)          &0.018(0.001)          &0.019(0.001) \\

\bottomrule[2pt]
\end{tabular}}
\end{table}
The simulation results are presented in Table \ref{table:1}, from which we mainly draw the following essential conclusions. Firstly, it confirms that distributed ECA is a more robust method compared with the distributed PCA. As revealed by the results,  our distributed ECA has a pretty good performance even when data is extremely heavy-tailed, while performs comparably with the D-PCA under Gaussian distribution. Secondly, the performance of distributed ECA is comparable with that of the F-DCA, which indicates that the distributed algorithm works as well in the case when we have all observations in one machine and coincides with our theoretical analysis.  Thus the distributed algorithm enjoys high computation efficiency while barely loses any accuracy.

\vspace{2em}

\begin{figure}[!h]
  \centering
  \includegraphics[width=16cm]{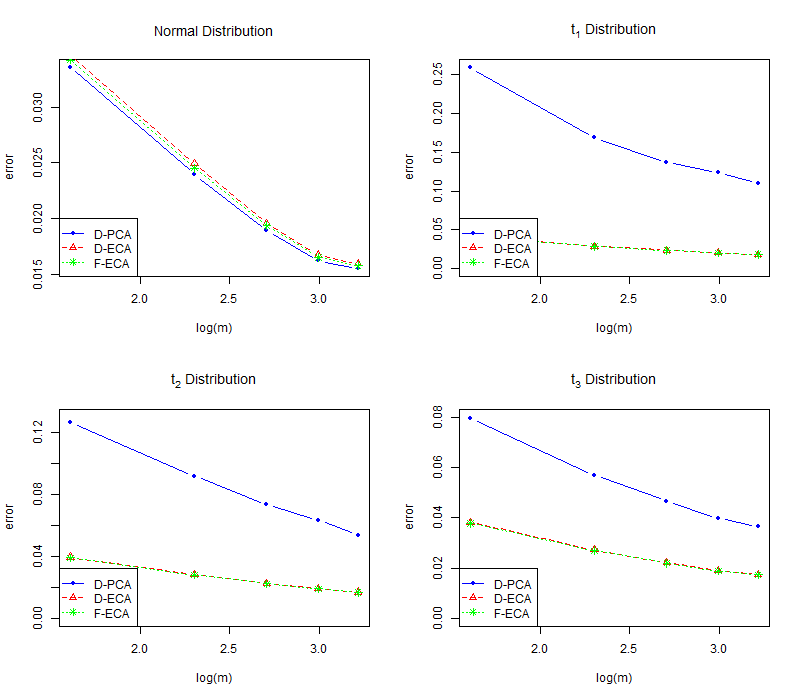}
 \caption{The  errors ( $\rho_1(\bL,\widehat\bL)$ ) versus the number of machines $m$ for distributed PCA, distributed ECA and full sample ECA under different distributions with fixed $p=20$.}\label{fig:2}
 \end{figure}

 Figure \ref{fig:1} demonstrates the decay rate of the error $\rho_1(\bL,\widehat\bL)$  as the  number of servers $m$ increases when $(n,p)=\{(200,20),(200,50)\}$, where $\widehat\bL$ denotes the  orthonormal output from the D-ECA . From Figure \ref{fig:1}, we can see that the relationship between $\log(\rho_1(\bL,\widehat\bL))$ and $\log(m)$ is linear. Then we consider a regression model as follows:
 \[
 \log(\rho_1(\bL,\widehat\bL))=\beta_0+\beta_1\log(m)+\epsilon.
\]
\par
The fitting result is that, $\hat{\beta}_1=-0.5019$ for $p=20$ and $\hat{\beta}_1=-0.49$ for $p=50$. As these experiments are conducted for fixed $n$,  the slope is consistent with the power of $m$ in the error bounds in Section 3, by noting that $\log(\rho_1(\bL,\widehat\bL))=\log(\rho(\bL,\widehat\bL))+\log(2K)/2$.

  Figure \ref{fig:2}  shows the errors of D-PCA, D-ECA and F-ECA  with   varying $m$  and fixed $p=20$ under different distributions. It is obvious from Figure \ref{fig:2} (see more simulation results in the supplement) that the error of D-PCA is much higher than that of D-ECA and F-ECA under the $t$-distribution and the differences increase as the degrees of freedom declines, which again demonstrates that distributed ECA is a much more robust method than distributed PCA. Meanwhile, the errors of D-ECA and F-ECA are very close. Thus we confirm the conclusion that distributed ECA is a computation efficient method without losing much accuracy.  In conclusion, when data are stored across different machines due to various kinds of reasons including communication cost,  privacy, data security, our D-ECA can be regarded as an alternative to ECA based on full-samples. Also, it performs much more robust than distributed PCA in \cite{fan2019distributed}.

\subsection{Real example}
 We apply the distributed algorithm for elliptical factor model to an macroeconomic real dataset, i.e., U.S. weekly economic data which are available at Federal Reserve Economic Data (\url{https://fred.stlouisfed.org/}). Our dataset consists of 40 variables $(p=40)$ from October 2nd, 1996, to April 11th, 2018 $(N=1124)$. We divide the dataset into four equal subsets and store them on four servers $(m=4)$. The series can roughly be classified into 5 groups: (1) Assets (series 1-3), (2) Loans (series 4-14), (3) Deposits and Securities, briefly denoted  as DAS (series 15-26), (4) Liabilities (series 27-32), (5) Currency (series 33-40).
By preliminary test, we find that more than two-thirds of the variables show the characteristics of heavy-tails, which indicates  ECA is more suitable than PCA in analyzing the dataset.

We compare the forecasting performance of factors obtained by D-ECA, F-ECA,  F-PCA, and D-PCA respectively. We adopt the rolling window prediction schemes which employs a fixed-length window of the most recent data (i.e., 10 weekly observations) to train the factor models and conduct  out-of-sample forecasting. To measure the forecasting performance of the extracted factors, we establish the $h$-step ahead prediction model:
\[
\hat x_{i,t+h}=\hat{\alpha}_h +\hat{\beta}_h\hat{\bm f}_t.
\]

 We only regress $\hat x_{i,t+h}$ on $\hat{\bm f}_t$ without considering the lagged variable of $\hat x_{i,t}$ as we find that including the lagged variables in the model is less effective,  which is consistent with the conclusion of \cite{stock2002macroeconomic}.

\begin{table}[htbp]
\caption{The 1-step to 4-step ahead forecast errors. The ratios are reported  for a more intuitive comparison of the prediction errors.}\label{table:3}
\renewcommand{\arraystretch}{1.6}
\centering
\scalebox{0.8}{\begin{tabular}{ccccccc}
\toprule[2pt]
$Step$ &  $Method$ & $Assets$  &$Loans$ & $DAS$& $Liabilities$ &$Currency$\\
\hline
\multirow{3}{*}{1}&D-ECA/F-ECA&0.98609&0.98390&1.00540& 0.98852 &1.00038\\
&D-ECA/F-PCA&0.98727&0.98329&1.00378&0.98944&0.99921\\
&F-PCA/F-ECA&1.00089&1.00081&1.00366&1.00124&1.00333\\
\hline
\multirow{3}{*}{2}&D-ECA/F-ECA&0.96390&0.98154&0.97463&0.97224 &0.97600\\
&D-ECA/F-PCA&0.96576&0.98136&0.97796&0.97125&0.97828\\
&F-PCA/F-ECA&0.99825&1.00355&0.98876&1.00359&0.99126\\
\hline
\multirow{3}{*}{3}&D-ECA/F-ECA&0.97880&0.98323&1.00612&0.94963&1.00575 \\
&D-ECA/F-PCA&0.97804&0.98261&1.00620&0.94734&1.00565\\
&F-PCA/F-ECA&1.00277&1.00819&1.00747&1.00701&1.00720\\
\hline
\multirow{3}{*}{4}&D-ECA/F-ECA&0.93355&0.96147& 0.96971&0.95796&0.96653\\
&D-ECA/F-PCA&0.93786&0.96351&0.94353&0.96158&0.96692\\
&F-PCA/F-ECA&0.99195&1.00563&0.99795&1.00233&0.99923\\

\bottomrule[2pt]
\end{tabular}}
\end{table}

 From Table \ref{table:3}, we see that the ahead prediction errors of D-ECA and F-ECA are similar. For 2-step ahead forecast and 4-step ahead forecast, the D-ECA prediction errors of these five groups are smaller than these of F-ECA. The prediction errors of D-ECA are not as good as those of F-ECA only for the 1-step ahead forecast and 3-step ahead forecast for DAS and Currency. On the other hand, whether distributed or full sample, the prediction errors of ECA is smaller than PCA. We also find that the F-ECA is better than the F-PCA in the 1-step ahead forecast and 3-step ahead forecast. The 2-step and 4-step prediction errors of Loans and Liabilities of F-ECA are better than those of F-PCA.

 Overall, the factors extracted by the distributed algorithms tend to have higher forecasting performance than the algorithms with the full data. When it is extremely difficult to aggregate large datasets due to the issues of communication cost, privacy concern, data security, ownership and  among many others, the distributed ECA algorithm provides a solution and may have higher forecasting performance than performing ECA/PCA to the whole data.

\section{Discussion}\label{section6}
	We proposed a robust distributed  algorithm to estimate principal eigenspace for high dimensional, elliptically distributed data, which are stored across multiple machines. The distributed algorithm reduce the transmission cost, and performs much better than distributed PCA when data are extremely heavy-tailed. The algorithm goes as follows: in the first step, we compute the $K$ leading eigenvectors of sample Kendall's tau matrix on each server and then transport them to the central server; in the second step, we take an average of the projection operators and perform a spectral decomposition on it to get the $K$ leading eigenvectors, which spans the final estimated eigenspace. We derive the  convergence rate of the robust distributed estimator. Numerical studies show that the proposed algorithm works comparably with the full sample ECA while performs better than distributed PCA when data are heavy-tailed. We also extend the algorithm to learn the elliptical factor model in a distributed manner, and the theoretical analysis is more challenging and we leave it as future work. It's also interesting to study the ``vertical" division for distributed robust PCA and elliptical factor model.

\section*{Acknowledgements}
The authors gratefully acknowledge National Science Foundation of China (12171282, 11801316), National Statistical Scientific Research Key Project (2021LZ09), Young Scholars Program of Shandong University, Project funded by
China Postdoctoral Science Foundation (2021M701997).

%\section{Supplementary Material}
%The technical proofs of the main results, extra empirical studies, and details of the data sets are included in the Supplementary Material.

\bibliographystyle{agsm}
\bibliography{Ref}

\setlength{\bibsep}{1pt}

\renewcommand{\baselinestretch}{1}
\setcounter{footnote}{0}
\clearpage
\setcounter{page}{1}
\title{
	\begin{center}
		\Large Supplementary Materials for ``Distributed Learning for Principle Eigenspaces without Moment Constraints"
	\end{center}
}
\date{}
\begin{center}
	\author{
	Yong He
		\footnotemark[1],
	Zichen Liu\footnotemark[1],
Yalin Wang\footnotemark[1]
	}
\renewcommand{\thefootnote}{\fnsymbol{footnote}}
\footnotetext[1]{Institute of Financial Studies, Shandong University, China; e-mail:{\tt heyong@sdu.edu.cn }}

\end{center}
\maketitle

\appendix
This document provides detailed proofs of the main paper and additional simulation results. We start with some auxiliary lemmas in Section \ref{sec:lemma} and the proof of main theorems is provided in Section \ref{sec:theo}. Section \ref{sec:C} provides the additional simulation results.

\vspace{1em}
\section{Auxiliary Lemmas}\label{sec:lemma}

\newtheorem{lem}{Lemma}[section]

\begin{lemma}\label{ECAlemma}
Let $\bX$ follows a continuous elliptical distribution, that is, $\bX\sim ED(g;\boldsymbol{\mu},\boldsymbol{\Sigma},\boldsymbol{\nu})$ with $\mathbb{P}(\xi=0)=0$. $\textbf{K}$ is the population multivariate Kendall's tau matrix, its rank is assumed to be $rank(\boldsymbol{\Sigma})=q$. Then:
\begin{equation}
\lambda_{j}(\textbf{K})=\mathbb{E}\left(\frac{\lambda_{j}(\boldsymbol{\Sigma})g_{j}}{\lambda_{1}(\boldsymbol{\Sigma})g_{1}^{2}+...+\lambda_{q}(\boldsymbol{\Sigma})g_{q}^{2}}\right)\notag,
\end{equation}
where $\boldsymbol{g}=(g_{1},\ldots,g_{q})^\top\sim\mathcal{N}(\textbf{0},\textbf{I}_{q})$ is a standard multivariate Gaussian distribution. In addition, $\textbf{K}$ and $\boldsymbol{\Sigma}$ share the same eigenspace with the same descending order of the eigenvalues.
\end{lemma}
\par
This lemma is the foundation of principal component analysis for elliptically distributed data, and the proof can be found in \cite{han2018eca}.

\begin{lemma}\label{auxlem1}
Under \textbf{Assumption 3.1} in the main paper, there exist a positive constant $c_0$ and a large $p_0$, for $p\geq p_0$:
$$\frac{1}{\lambda_{K}(\textbf{K})-\lambda_{K+1}(\textbf{K})}\leq c_0 \frac{\text{Tr}(\boldsymbol{\Sigma})}{\lambda_{K}(\boldsymbol{\Sigma})-\lambda_{K+1}(\boldsymbol{\Sigma})}=c_0\kappa r$$
\end{lemma}

\begin{proof}
First, we claim that when $\gamma=\lambda_{1}(\boldsymbol{\Sigma})/\lambda_{p}(\boldsymbol{\Sigma})\leq Cp^{\alpha}$, $$\Vert \boldsymbol{\Sigma} \Vert_{F}\sqrt{\log p}=o(\text{Tr}(\boldsymbol{\Sigma})).$$
It is obviously true by:$$\text{Tr}(\boldsymbol{\Sigma})\geq p\lambda_{p}\geq p^{1-\alpha} \frac{\lambda_{1}}{C}=p^{1-\alpha} \Vert \boldsymbol{\Sigma} \Vert_{2}/C\geq p^{1/2-\alpha}\Vert \boldsymbol{\Sigma} \Vert_{F}$$
By Theorem 3.5 in \cite{han2018eca}, for large $p$,
\begin{equation}
\begin{aligned}
\lambda_{K}(\textbf{K})-\lambda_{K+1}(\textbf{K})\geq\frac{\lambda_{K}(\boldsymbol{\Sigma})}{\text{Tr}(\boldsymbol{\Sigma})+4\Vert \boldsymbol{\Sigma} \Vert_{F}\sqrt{\log p}+8\Vert \boldsymbol{\Sigma} \Vert_{2}\log p}(1-\frac{\sqrt{3}}{p^{2}})\\
-\frac{\lambda_{K+1}(\boldsymbol{\Sigma})}{\text{Tr}(\boldsymbol{\Sigma})-4\Vert \boldsymbol{\Sigma} \Vert_{F}\sqrt{\log p}}-\frac{1}{p^4}:= f_1(p) \frac{\lambda_{K}(\boldsymbol{\Sigma})}{\text{Tr}(\boldsymbol{\Sigma})}- f_{2}(p) \frac{\lambda_{K+1}(\boldsymbol{\Sigma})}{\text{Tr}(\boldsymbol{\Sigma})}-\frac{1}{p^4}\notag.
\end{aligned}
\end{equation}
According to the claim above, we know that $f_1,f_2\rightarrow1$ as $p\rightarrow \infty$. Also,
$$\frac{f_{1}(p)}{f_{2}(p)}=\frac{1-\sqrt{3}/p^{2}}{\text{Tr}(\boldsymbol{\Sigma})+4\Vert \boldsymbol{\Sigma} \Vert_{F}\sqrt{\log p}+8\Vert \boldsymbol{\Sigma} \Vert_2\log p}(\text{Tr}(\boldsymbol{\Sigma})-4\Vert \boldsymbol{\Sigma} \Vert_{F}\sqrt{\log p})<1.$$
By the assumption on ratio of eigenvalues, we have:

\begin{equation}
\begin{aligned}
&f_1(p) \frac{\lambda_{K}(\boldsymbol{\Sigma})}{\text{Tr}(\boldsymbol{\Sigma})}- f_{2}(p) \frac{\lambda_{K+1}(\boldsymbol{\Sigma})}{\text{Tr}(\boldsymbol{\Sigma})}-\frac{1}{p^4}\\
&= \frac{f_1(p)}{2} \frac{\lambda_{K}(\boldsymbol{\Sigma})}{\text{Tr}(\boldsymbol{\Sigma})}- \frac{f_{1}(p)}{2} \frac{\lambda_{K+1}(\boldsymbol{\Sigma})}{\text{Tr}(\boldsymbol{\Sigma})}+\left(\frac{f_1(p)}{2} \frac{\lambda_{K}(\boldsymbol{\Sigma})}{\text{Tr}(\boldsymbol{\Sigma})}+\frac{f_{1}(p)}{2} \frac{\lambda_{K+1}(\boldsymbol{\Sigma})}{\text{Tr}(\boldsymbol{\Sigma})}-f_{2}(p)\frac{\lambda_{K+1}(\boldsymbol{\Sigma})}{\text{Tr}(\boldsymbol{\Sigma})}\right)-\frac{1}{p^4}\\
&\geq \frac{f_1(p)}{2} \frac{\lambda_{K}(\boldsymbol{\Sigma})}{\text{Tr}(\boldsymbol{\Sigma})}- \frac{f_{1}(p)}{2} \frac{\lambda_{K+1}(\boldsymbol{\Sigma})}{\text{Tr}(\boldsymbol{\Sigma})}+\frac{\lambda_{K+1}(\boldsymbol{\Sigma})}{\text{Tr}(\boldsymbol{\Sigma})}(f_{1}(p)(1+\delta)-f_{2}(p))-\frac{1}{p^4}
\\
&\geq \frac{1}{4}\frac{\lambda_{K}(\boldsymbol{\Sigma})-\lambda_{K+1}(\boldsymbol{\Sigma})}{\text{Tr}(\boldsymbol{\Sigma})}-\frac{1}{p^{4}}(\text{for $p$ large enough})\geq \frac{1}{c_0}\frac{\lambda_{K}(\boldsymbol{\Sigma})-\lambda_{K+1}(\boldsymbol{\Sigma})}{\text{Tr}(\boldsymbol{\Sigma})}\notag,
\end{aligned}
\end{equation}
where the last inequality holds because:$$\frac{\text{Tr}(\boldsymbol{\Sigma})}{\lambda_{K}(\boldsymbol{\Sigma})-\lambda_{K+1}(\boldsymbol{\Sigma})}\leq \frac{p\lambda_{1}}{\delta\lambda_{K+1}}\leq p\frac{\lambda_{1}}{\delta\lambda_{p}}=O(p^{1+\alpha})=o(p^{4}).$$
Consequently, there exists a large $p_0$, such that when $p\geq p_0$: $$\frac{1}{\lambda_{K}(\textbf{K})-\lambda_{K+1}(\textbf{K})}\leq c_0 \frac{\text{Tr}(\boldsymbol{\Sigma})}{\lambda_{K}(\boldsymbol{\Sigma})-\lambda_{K+1}(\boldsymbol{\Sigma})}=c_0\kappa r.$$
\end{proof}

\begin{lemma}\label{auxlem2}
When $\Vert \boldsymbol{\Sigma}^{*}-\boldsymbol{V}_{K}\boldsymbol{V}_{K}^{T} \Vert_{2}\leq \frac{1}{4}$, we have:
$$\frac{1}{2}\leq\lambda_{K}(\boldsymbol{\Sigma}^{*})-\lambda_{K+1}(\boldsymbol{\Sigma}^{*})\leq \frac{3}{2}.$$
\end{lemma}
\begin{proof}
Since $\bV_{K}^{T}\bV_{K}=\Ib_{K}$, we have:$$\lambda_{K}(\bV_{K}\bV_{K}^{T})=1,\lambda_{K+1}(\boldsymbol{V}_{K}\boldsymbol{V}_{K}^{T})=0.$$
By Weyl theorem, $$\lambda_{K+1}(\boldsymbol{\Sigma}^{*})-\lambda_{K+1}(\boldsymbol{V}_{K}\boldsymbol{V}_{K}^{T})\leq \Vert \boldsymbol{\Sigma}^{*}-\boldsymbol{V}_{K}\boldsymbol{V}_{K}^{T}\Vert_{2}\qquad\lambda_{K+1}(\boldsymbol{\Sigma}^{*})\leq \frac{1}{4}.$$
Similarly, $$\lambda_{K}(\boldsymbol{\Sigma}^{*})\leq\frac{5}{4},\lambda_{K}(\boldsymbol{\Sigma}^{*})\geq\frac{3}{4},\lambda_{K+1}(\boldsymbol{\Sigma}^{*})\geq-\frac{1}{4}$$
Hence,$$\frac{1}{2}\leq\lambda_{K}(\boldsymbol{\Sigma}^{*})-\lambda_{K+1}(\boldsymbol{\Sigma}^{*})\leq \frac{3}{2}.$$
\end{proof}

\begin{lemma}\label{ecalemma}
Let $k(\cdot)$ : $\mathcal{X}\times\mathcal{X}\rightarrow\mathbb{R}^{d\times d}$ be a matrix value kernel function. Let $\bX_1,\ldots,\bX_n$ be $n$ independent observations of an random variable $\bX\in\mathcal{X}$. Suppose that, for any $i\neq i^{'}\in\left\{1,\ldots,n\right\}$, $\mathbb{E}k(\bX_i,\bX_{i^{'}})$ exists and there exist two constants $R_1,R_2>0$ such that
$$
\Vert k(\bX_{i},\bX_{i^{'}})-\mathbb{E}k(\bX_{i},\bX_{i'})\Vert_{2}\leq R_{1}\ \text{and}\ \Vert\mathbb{E}\{k(\bX_{i},\bX_{i'})-\mathbb{E}k(\bX_{i},\bX_{i'})\}^{2}\Vert_{2}\leq R_2
$$
We then have
$$
\mathbb{P}\left(\left\Vert \frac{1}{\binom{n}{2}}\sum_{i<i'}k(\bX_{i},\bX_{i'})-\mathbb{E}k(\bX_{i},\bX_{i'})\right\Vert_{2}\geq t\right)\leq d\exp\left(-\frac{(n/4)t^{2}}{R_{2}+R_{1}t/3}\right).
$$
\end{lemma}
The proof of this lemma can be found in \cite{han2018eca}.

\begin{lemma}\label{lemma}
$\bX_{1},\ldots,\bX_{n}\in\mathbb{R}^{p}$ are samples from some elliptical distribution. Population and sample Kendall's tau matrix are denoted by $\textbf{K}$ and $\widehat{\textbf{K}}$, respectively. If $\log p=o(n)$, we have:
\begin{equation}\label{lemma1}
\Big\Vert \Vert \widehat{\textbf{K}}-\textbf{K}\Vert_{2} \Big\Vert_{\psi_1} \lesssim \sqrt{\frac{\log\,p}{n}}\notag.
\end{equation}
\end{lemma}
\begin{proof}

It is easily seen that $\text{Tr}(\textbf{K})=\text{Tr}(\widehat{\textbf{K}})=1$. Since $\textbf{K}$ and $\widehat{\textbf{K}}$ are positive semidefinite matrices, we have $\Vert \textbf{K} \Vert_{2}\leq 1,\Vert\widehat{\textbf{K}}\Vert_{2} \leq 1$, and thus $\Vert \textbf{K}-\widehat{\textbf{K}}\Vert_{2}\leq 2$.\par
Note that kernel function of Kendall's tau matrix is:
$$k(\bX_{i},\bX_{i'}):=\frac{(\bX_{i}-\bX_{i'})(\bX_{i}-\bX_{i'})^{T}}{\Vert \bX_{i}-\bX_{i'}\Vert_{2}^{2}}.
$$
By some simple algebra,
\begin{equation}
\begin{aligned}
\Vert \frac{(\bX_{i}-\bX_{i'})(\bX_{i}-\bX_{i'})^{T}}{\Vert \bX_{i}-\bX_{i'}\Vert_{2}^{2}}-\textbf{K}\Vert_{2}&\leq\Vert \frac{(\bX_{i}-\bX_{i'})(\bX_{i}-\bX_{i'})^{T}}{\Vert \bX_{i}-\bX_{i'}\Vert_{2}^{2}}\Vert_{2}+\Vert \textbf{K}\Vert_{2}\\
&=\text{Tr}(\frac{(\bX_{i}-\bX_{i'})(\bX_{i}-\bX_{i'})^{T}}{\Vert \bX_{i}-\bX_{i'}\Vert_{2}^{2}})+\Vert \textbf{K} \Vert_{2}=1+\Vert \textbf{K}\Vert_{2}\\
\Vert\mathbb{E}\{k(\bX_{i},\bX_{i'})-\mathbb{E}k(\bX_{i},\bX_{i'})\}^{2}\Vert_{2}&\leq\Vert \textbf{K}\Vert_{2}+\Vert \textbf{K}\Vert_{2}^{2}.
\notag
\end{aligned}
\end{equation}
According to Lemma \ref{ecalemma}, let $R_1=2,R_2=2$, then
\begin{equation}
\mathbb{P}(\Vert \widehat{\textbf{K}}-\textbf{K} \Vert_{2}\geq t) \leq p\,\exp\left(-\frac{(n/4)t^{2}}{2+2t/3}\right),\forall t \geq 0\notag
\end{equation}
We need to find a function $f(n,p)$ with respect to $n$ and $p$, which satisfies $f(n,p)\leq 2$ and
\begin{equation}\label{func}
p\,\exp\left(-\frac{(n/4)t^{2}}{2+2t/3}\right)\leq \exp(1-\frac{t}{f(n,p)}), \forall t\in\left[f(n,p),2\right]
\end{equation}
It is equivalent to:
\begin{equation}
\frac{1}{f(n,p)}\leq \mathop{\min}\limits_{t\in\left[f(n,p),2\right]}\left(\frac{1-\log\,p}{t}+\frac{nt/4}{2+2t/3}\right)\nonumber
\end{equation}
It is easy to see that the right term is an increasing function of $t$, so the last inequality above is equivalent to:
\begin{equation}
\left\{
\begin{aligned}
\frac{1}{f(n,p)} &\leq \frac{1-\log\,p}{f(n,p)}+\frac{nf(n,p)/4}{2+2f(n,p)/3}\\
f(n,p) &\leq 2
\end{aligned}
\right.\nonumber
\end{equation}
Let $f(n,p)=(\sqrt{\frac{36\log\,p}{n}})\wedge 2$,
it is easily to verify that it satisfies \eqref{func}.
Therefore,
\begin{equation}
\mathbb{P}(\Vert \widehat{\textbf{K}}-\textbf{K} \Vert_{2}\geq t) \leq exp(1-\frac{1}{(\sqrt{\frac{36\log\,p}{n}})\wedge 2})),\forall t \geq 0\notag
\end{equation}
By (5.14) in \cite{vershynin2010introduction}, we have
\begin{equation}
\Big\Vert \Vert \widehat{\textbf{K}}-\textbf{K} \Vert_{2} \Big\Vert_{\psi_1} \leq (\sqrt{\frac{36\log\,p}{n}})\wedge 2 \lesssim \sqrt{\frac{\log p}{n}}\notag
\end{equation}
\end{proof}

\section{Proof of the Main Theorems}\label{sec:theo}

\subsection*{Proof of Theorem 4.1}

\begin{proof}
According to Davis-Kahan theorem in \cite{yu2015useful}, we have:
\begin{equation}
\boldsymbol{\rho}(\widetilde{\bV}_{K},\bV_{K}^{*}) \lesssim \frac{\Vert \widetilde{\boldsymbol{\Sigma}}-\boldsymbol{\Sigma}^{*}\Vert_{F}}{\lambda_{K}(\boldsymbol{\Sigma}^{*})-\lambda_{K+1}(\boldsymbol{\Sigma}^{*})}\notag
\end{equation}

\begin{equation}
\begin{aligned}
\Big\Vert \boldsymbol{\rho}(\widetilde{\bV}_{K},\bV_{K}^{*}) \Big\Vert_{\psi_1} \leq \sqrt{2} \Big\Vert \Vert& \widetilde{\boldsymbol{\Sigma}}-\boldsymbol{\Sigma}^{*}\Vert_{F} \Big\Vert_{\psi_1}/\delta=\sqrt{2}\Big\Vert \Vert \frac{1}{m} \sum_{i=1}^{m}\bV_{K}^{(i)}\bV_{K}^{(i)T}-\boldsymbol{\Sigma}^{*}\Vert_{F} \Big\Vert_{\psi_1}/\delta \leq\\
&\sqrt{2}\Big\Vert \frac{1}{m} \sum_{i=1}^{m}\Vert \bV_{K}^{(i)}\bV_{K}^{(i)T}-\boldsymbol{\Sigma}^{*}\Vert_{F} \Big\Vert_{\psi_1}/\delta\notag
\end{aligned}
\end{equation}
Note that $\Vert \bV_{K}^{(i)}\bV_{K}^{(i)T}-\boldsymbol{\Sigma}^{*} \Vert_{F} \leq k+\Vert \boldsymbol{\Sigma}^{*} \Vert_{F}<\infty$ is a bounded random variable, thus it is sub-Gaussian, naturally with finite Orlizc $\psi_1$ norm. Hence according to \textbf{Lemma 4} in \cite{fan2019distributed},
\[
\begin{aligned}
\Big\Vert \frac{1}{m} \sum_{i=1}^{m}\Vert \bV_{K}^{(i)}\bV_{K}^{(i)T}-\boldsymbol{\Sigma}^{*}\Vert_{F} \Big\Vert_{\psi_1} &\lesssim \sqrt{\frac{\sum_{i=1}^{m}(\Big\Vert \Vert \widehat{\bV}_{K}^{(i)}\widehat{\bV}_{K}^{(i)T}-\boldsymbol{\Sigma}^{*}\Vert_{F} \Big\Vert_{\psi_1})^{2}}{m^{2}}}\\
&\leq\frac{1}{\sqrt{m}}\max\limits_{l}\Big\Vert \Vert \widehat{\bV}_{K}^{(l)}\widehat{\bV}_{K}^{(l)T}-\boldsymbol{\Sigma}^{*}\Vert_{F} \Big\Vert_{\psi_1}.\\
\end{aligned}
\]
In fact, we have
\begin{equation}
\begin{aligned}
\Big\Vert \Vert \widehat{\bV}_{K}^{(i)}\widehat{\bV}_{K}^{(i)T}-\boldsymbol{\Sigma}^{*}\Vert_{F} \Big\Vert_{\psi_1}&\leq \Big\Vert \Vert \widehat{\bV}_{K}^{(i)}\widehat{\bV}_{K}^{(i)T}-\bV_{K}\bV_{K}^{T}\Vert_{F} \Big\Vert_{\psi_1}+\Vert \boldsymbol{\Sigma}^{*}-\bV_{K}\bV_{K}^{T}\Vert_{F}\\
&\leq\Vert \boldsymbol{\rho}(\widehat{\bV}_{K}^{(i)},\bV_{K})\Vert_{\psi_1}+\mathbb{E}(\Vert \widehat{\bV}_{K}^{(i)}\widehat{\bV}_{K}^{(i)T}-\bV_{K}\bV_{K}^{T}\Vert_{F})\\
&\leq 2\Vert \boldsymbol{\rho}(\bV_{K},\widehat{\bV}_{K}^{(i)})\Vert_{\psi_1}\notag, \forall i\in [m],
\end{aligned}
\end{equation}
where the second inequality is by Jensen's inequality.
Combining
\begin{equation}
\boldsymbol{\rho}(\bV_{K},\widehat{\bV}_{K}^{(i)})=\Vert \widehat{\bV}_{K}^{(i)}\widehat{\bV}_{K}^{(i)T}-\bV_{K}\bV_{K}^{T}\Vert_{F}=\sqrt{2}\Vert sin\Theta(\widehat{\bV}_{K}^{(i)},\bV_{K})\Vert_{F}\notag,
\end{equation}
and the generalized \textbf{Davis-Kahan} theorem in \cite{yu2015useful}, which is
\begin{equation}
\Vert sin\Theta(\widehat{\bV}_{K}^{(i)},\bV_{K})\Vert_{F} \leq \frac{2\sqrt{K}\Vert \widehat{\textbf{K}}-\textbf{K}\Vert_{2}}{\lambda_{K}(\textbf{K})-\lambda_{K+1}(\textbf{K})}\notag,
\end{equation}
we have
\begin{equation}
\Vert \boldsymbol{\rho}(\bV_{K},\widehat{\bV}_{K}^{(i)})\Vert_{\psi_1}\leq \frac{2\sqrt{2K}}{\lambda_{K}(\textbf{K})-\lambda_{K+1}(\textbf{K})}\Big\Vert \Vert \widehat{\textbf{K}}-\textbf{K}\Vert_{2}\Big\Vert_{\psi_1}\notag,
\end{equation}
Hence, combining all of the equations above and Lemma \ref{lemma}, we have
\begin{equation}\label{equ}
\Vert \boldsymbol{\rho}(\widehat{\bV}_{K},\bV_{K}^{*})\Vert_{\psi_1}\leq 48\sqrt{\frac{K\log\,p}{m\cdot n}}\frac{1}{\delta\cdot(\lambda_{K}(\textbf{K})-\lambda_{K+1}(\textbf{K}))}.
\end{equation}
According to Lemma \ref{auxlem1} and equation \eqref{equ}, it is not hard to show that
\begin{equation}
\Vert \boldsymbol{\rho}(\widehat{\bV}_{K},\bV_{K}^{*})\Vert_{\psi_1}\lesssim \kappa r K^{3/2}(\frac{\log p}{N})^{1/2}\notag.
\end{equation}
\end{proof}

\subsection*{Proof of Theorem 4.2}

\begin{proof}
Define $\textbf{E}=\widehat{\textbf{K}}-\textbf{K}$, $\Delta=\lambda_{K}(\textbf{K})-\lambda_{K+1}(\textbf{K})$. According to Theorem 3 in \cite{fan2019distributed}, we have
\begin{equation}
\begin{aligned}
&\boldsymbol{\rho}(\bV_{K}^{*},\bV_{K}) \leq \Vert \mathbb{E}(\widehat{\bV}_{K}^{(1)}\widehat{\bV}_{K}^{(1)T})-\bV_{K}\bV_{K}^{T} \Vert_{F}\lesssim \sqrt{K}\mathbb{E}(\frac{\Vert \textbf{E} \Vert_{2}}{\Delta})^{2}\\
&\lesssim \sqrt{K}\Delta^{-2}\mathbb{E}\Big\Vert \Vert \textbf{E} \Vert_{2} \Big\Vert_{\psi_1}^{2}\lesssim \kappa^{2}r^{2}K^{5/2}\frac{m\log p}{N}\notag.
\end{aligned}
\end{equation}
The final inequality is derived from Lemma \ref{lemma}.
\end{proof}
%\begin{remark}
%Recall what we need in Lemma \ref{auxlem2}, as a corollary, proof of Theorem \ref{thm4.2} gives a sufficient condition on $\Vert \boldsymbol{\Sigma}^{*}-\boldsymbol{V}_{K}\boldsymbol{V}_{K}^{T} \Vert_{2}\leq 1/4$¡£
%\end{remark}

\section{Additional Simulation results}\label{sec:C}
Figure \ref{fig:3}  show the errors of D-PCA, D-ECA and F-ECA  with   varying $m$  and and fixed $p=50$ under different distributions. It is obvious from the Figure \ref{fig:3}  that the error of D-PCA is much higher than that of D-ECA and F-ECA under the $t$-distribution and the differences increase as the degrees of freedom declines, which again proves distributed ECA is a much more robust method than distributed PCA. Meanwhile, the errors of D-ECA and F-ECA are very close. Thus we confirm the conclusion that distributed ECA is a computation efficient method without losing much accuracy.  In conclusion, when data are stored across different machines due to various kinds of reasons including communication cost,  privacy, data security, our D-ECA can be regarded as an alternative to ECA based on full-samples.
 \begin{figure}[!h]
  \centering
  \includegraphics[width=16cm]{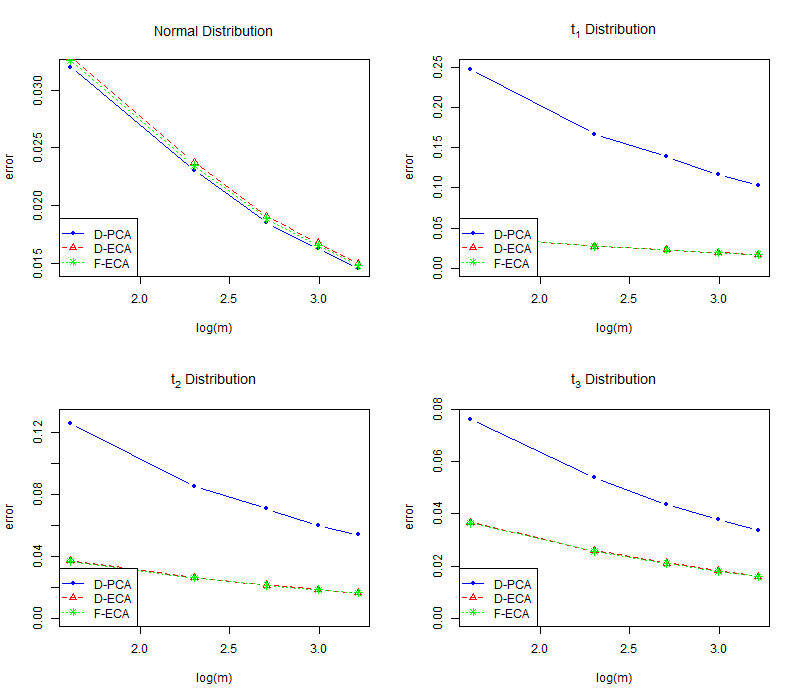}
 \caption{The  errors ($\rho_1(\bL,\widehat\bL)$) versus the number of machines $m$ for distributed PCA, distributed ECA and full sample ECA under different distributions with fixed $p=50$.}\label{fig:3}
 \end{figure}
\end{document}